\documentclass[11pt]{article}
\pdfoutput=1
\usepackage[margin=1in]{geometry}
\usepackage[onehalfspacing]{setspace}
\usepackage{xcolor}
\usepackage{amsmath}
\usepackage{amsthm}
\usepackage{hyperref}
\usepackage[nameinlink,capitalize]{cleveref}
\usepackage{graphicx,amssymb}
\usepackage{libertine}

\newcommand{\citet}[1]{\cite{{#1}}}
\newcommand{\citep}[1]{\cite{{#1}}}

\title{Approximately Efficient Bilateral Trade with Samples}

\date{}

\newtheorem{theorem}{Theorem}[section]
\newtheorem{lemma}[theorem]{Lemma}

\newtheorem{proposition}[theorem]{Proposition}
\theoremstyle{definition}
\newtheorem{definition}[theorem]{Definition}

\crefname{Program}{Program}{Programs}
\creflabelformat{Program}{(#2\textup{#1})#3}

\definecolor{linkc}{rgb}{0.1, 0.5, 0.7}
\definecolor{citec}{rgb}{0.6, 0.3, 0.7}
\definecolor{urlc}{rgb}{0.5, 0.1, 0.2}
\hypersetup{
    colorlinks=true,
    linkcolor=linkc,
    citecolor=citec,
    urlcolor=urlc
}

\newcommand{\notshow}[1]{{}}
\newcommand{\AutoAdjust}[3]{{ \mathchoice{ \left #1 #2  \right #3}{#1 #2 #3}{#1 #2 #3}{#1 #2 #3} }}
\newcommand{\Xcomment}[1]{{}}

\newcommand{\InParentheses}[1]{\AutoAdjust{(}{#1}{)}}
\newcommand{\InBrackets}[1]{\AutoAdjust{[}{#1}{]}}

\DeclareMathOperator*{\argmax}{arg\,max}

\DeclareMathOperator*{\E}{\mathbb{E}}

\newcommand{\indic}[1]{\mathbb{I}\InBrackets{#1}}

\newcommand{\given}{\,\vert\,}

\newcommand{\gft}{\mathtt{GFT}}
\newcommand{\firstbest}{\mathtt{FB}}
\newcommand{\sellerp}{\mathtt{SellerP}}
\newcommand{\buyerp}{\mathtt{BuyerP}}
\newcommand{\sellersample}{\mathtt{SellerSample}}
\newcommand{\buyersample}{\mathtt{BuyerSample}}
\newcommand{\srev}{\mathtt{SPro}}
\newcommand{\suti}{\mathtt{SUti}}
\newcommand{\brev}{\mathtt{BPro}}
\newcommand{\buti}{\mathtt{BUti}}
\newcommand{\wel}{\mathtt{Wel}}
\newcommand{\rev}{\mathtt{Rev}}
\newcommand{\emprev}{\widetilde{\mathtt{Rev}}}

\newcommand{\dist}{F}

\newcommand{\seller}{S}
\newcommand{\valS}{\upsilon_s}
\newcommand{\distS}{F_s}
\newcommand{\empdistS}[1]{\tilde{F}_s^{(#1)}}
\newcommand{\val}{\upsilon}
\newcommand{\prob}{p}

\newcommand{\valB}{\upsilon_b}
\newcommand{\distB}{F_b}

\renewcommand{\ref}{\textcolor{red}{Always cref!}} 

\begin{document}

\begin{titlepage}
\author{Yuan Deng \\ \normalsize Google Research \\ \normalsize\href{mailto:dengyuan@google.com}{dengyuan@google.com} \\
  \and Jieming Mao \\ \normalsize Google Research \\ \normalsize\href{mailto:maojm@google.com}{maojm@google.com} \\
  \and Balasubramanian Sivan \\ \normalsize Google Research \\ \normalsize\href{mailto:balusivan@google.com}{balusivan@google.com} \\
  \and Kangning Wang \\ \normalsize Rutgers University \\ \normalsize\href{mailto:kn.w@rutgers.edu}{kn.w@rutgers.edu} \\
  \and Jinzhao Wu \\ \normalsize Yale University \\ \normalsize\href{mailto:jinzhao.wu@yale.edu}{jinzhao.wu@yale.edu}
}

\maketitle
\thispagestyle{empty}
\setcounter{page}{0}

\begin{abstract}
We study the social efficiency of bilateral trade between a seller and a buyer. In the classical Bayesian setting, the celebrated Myerson–Satterthwaite impossibility theorem states that no Bayesian incentive-compatible, individually rational, and budget-balanced mechanism can achieve full efficiency. As a counterpoint, Deng, Mao, Sivan, and Wang (STOC 2022) show that if pricing power is delegated to the right person (either the seller or the buyer), the resulting mechanism can guarantee at least a constant fraction of the ideal (yet unattainable) gains from trade.

In practice, the agent with pricing power may not have perfect knowledge of the value distribution of the other party, and instead may rely on samples of that distribution to set a price. We show that for a broad class of sampling and pricing behaviors, the resulting market still guarantees a constant fraction of the ideal gains from trade in expectation. Our analysis hinges on the insight that social welfare under sample-based pricing approximates the seller’s optimal revenue---a result we establish via a reduction to a random walk.
\end{abstract}

\end{titlepage}

\section{Introduction}
\subsection{Approximately Efficient Bilateral Trade}
In various economies, sellers and buyers exchange goods for money; such a process is referred to as \emph{bilateral trade}. In a simple and canonical bilateral trade setting, a seller has one indivisible item for sale, and a buyer is interested in acquiring it. If the item is sold, the seller incurs a cost of $\valS$, which represents either a production cost or the disutility of parting with the item. The buyer, in turn, derives a value of $\valB$ from obtaining the item.

We aim to design mechanisms that facilitate such trades and, first, it helps to describe the classical Bayesian model of bilateral trade. In this model, the seller's cost $\valS$ and the buyer's value $\valB$ are independently drawn from known distributions $\distS$ and $\distB$, respectively. Although these distributions $\distS$ and $\distB$ are common knowledge, the actual realized values $\valS$ and $\valB$ remain private to the seller and the buyer, respectively. They interact through a mechanism and have standard quasilinear utilities:
\begin{itemize}
    \item The utility of the buyer is given by his value $\valB$ multiplied by the probability of obtaining the item, minus his expected payment.
    \item The utility of the seller is given by the expected payment she received, minus her cost $\valS$ multiplied by the probability of losing the item.
\end{itemize}
A valid mechanism must satisfy three standard conditions in the field of mechanism design:
\begin{itemize}
    \item \emph{Bayesian Incentive-Compatibility} (BIC): Neither the seller nor the buyer can improve their utility by misreporting their private information.
    \item \emph{Individual Rationality} (IR): Neither the seller nor the buyer should receive negative utility by participating in the mechanism.
    \item \emph{Budget Balance} (BB): The mechanism cannot subsidize trades. In other words, the buyer's payment to the mechanism must at least cover the payment to the seller.
\end{itemize}

Our goal is to design a trading mechanism that maximizes social efficiency. From a utilitarian perspective, if the buyer values an item more than the seller, then a successful trade between them will improve social efficiency. Quantitatively, a trade between a buyer with value $\valB$ and a seller with cost $\valS$ contributes a net utility of $\valB - \valS$ to society. Ideally, a bilateral trade mechanism should guarantee full efficiency, that is, trade should occur whenever the realized value $\valB \sim \distB$ is higher than the realized cost $\valS \sim \distS$.

However, in general, no mechanism that satisfies the three standard conditions can guarantee full efficiency. This fundamental limitation of efficient trading, known as the Myerson--Satterthwaite impossibility theorem \cite{myerson1983efficient}, is a landmark result in the field of mechanism design and a key reference in the 2007 Nobel Prize in Economics.

Influenced by the growing interdisciplinary collaboration in the field of economics and computation, researchers have explored the topic of bilateral trade efficiency through the lens of competitive analysis. Indeed, the \emph{gains from trade} in an ideal situation can be quantified as
\[
\E_{\valB \sim \distB, \valS \sim \distS}[\max(\valB - \valS, 0)],
\]
which we refer to as the \emph{first best}. A natural question is whether a mechanism can guarantee at least a constant fraction of the first best.

The work of \cite{mcafee2008gains} was the first to examine efficient trading mechanisms in this framework and showed that the fixed-price mechanism can guarantee a $\frac{1}{2}$-fraction of the first best when the median of $\distB$ is higher than the median of $\distS$. The work of \cite{DBLP:conf/stoc/DengMSW22} was the first to prove it without additional assumptions, showing that a mechanism can achieve at least a constant fraction of the first best regardless of the distributions $\distB$ and $\distS$. Specifically, it turns out that delegated pricing is sufficient; here, either the buyer or the seller has the power to set the price to maximize their own expected utility. Researchers have since shown that, for any prior distributions $\distB$ and $\distS$, one of these two delegation mechanisms can guarantee at least a $\frac{1}{3.15}$-fraction of the first best \cite{DBLP:conf/wine/Fei22,DBLP:conf/stoc/DengMSW22}.

\subsection{Pricing with Sample Access}
In real-world applications, the entity setting the price often lacks perfect knowledge of the other party's value or cost distribution. Instead, it typically relies on samples from the underlying distribution and uses the observed empirical distribution as a proxy of the true prior. A newcomer to an economy might have only a handful of samples, while a powerful entity in today's data-driven world may have access to millions. Furthermore, the price-setting entity can base its decision only on existing samples, or it can adopt an adaptive strategy and collect additional samples until it determines a price. We admit a broad range of possible behaviors of the price-setting entity, but assume that ultimately, this entity will choose a price to approximately optimize its utility with respect to the empirical distribution formed by samples it has observed so far.

Pricing with sample access has been extensively studied in the literature. In standard pricing and auction settings without additional assumptions, sellers with such behaviors typically fail to achieve a constant approximation to optimal revenue. Likewise, in these settings, the economy as a whole typically fails to achieve a constant approximation to optimal social welfare.

In this article, we explore the extent to which pricing with sample access differs from pricing with full knowledge of the prior distributions in bilateral trade. We ask specifically:
\begin{quote}
Is pricing with sample access in bilateral trade approximately efficient, that is, does it guarantee a constant fraction of gains from trade compared to the first-best outcome?
\end{quote}

\subsection{Our Results}

We answer this question in the affirmative. In our main result (\cref{thm:main} with $c = 1$), we consider the situation in which we delegate the pricing power to the seller or the buyer who only has sample access to the other's prior distribution and can employ a wide range of strategies. We show that the resulting mechanism guarantees at least a $\frac{1}{25.2}$-fraction of the first-best gains from trade, as long as the agent with pricing power chooses to optimize its empirical profit.

A key step in the proof of this result is the observation below (\cref{lem:general_prob} with $\delta = \frac{1}{4}$).
\begin{quote}
Consider a seller (with cost of $0$) who has access to samples of the buyer's value distribution and chooses to optimize her empirical revenue. It always holds in expectation that the resulting \emph{welfare} (with respect to the buyer's true value distribution) is at least a $\frac{1}{8}$-fraction of the optimal \emph{revenue} (again, with respect to the buyer's true value distribution).
\end{quote}
Note that we are showing that the resulting \emph{welfare} can approximate the optimal \emph{revenue}. This observation is in contrast to the facts that (1) the resulting welfare cannot guarantee a constant fraction of the public-prior-case welfare, let alone the optimal welfare, and (2) the resulting revenue cannot guarantee a constant fraction of the public-prior-case revenue, which is exactly the optimal revenue. (See \cref{appendix:wwrr} for proofs.) We think that this observation can be of independent interest. To prove it, we formulate our problem in the language of random walks and provide characterizations of its worst-case instances. The details are in \cref{subsec:worst_case_shape,subsec:integer_supp,subsec:general_supp}.

Finally, we consider the robustness of our results in the full version of \cref{thm:main}, and show that the constant-factor approximation is robust to alternative pricing strategies. Specifically, even if the price-setting entity only \emph{approximately} optimizes its empirical profit, our constant-factor efficiency guarantee still holds.
 
\subsection{Further Related Work}
\paragraph{Bilateral Trade.} There is a large body of literature on bilateral trade since the seminal work of \citet{myerson1983efficient}. We do not intend to be comprehensive, but instead we mention some of the works that are most related to ours. 

The fixed-price mechanism is one that announces a price $p$ and lets the seller and the buyer trade at price $p$ if the buyer's value is higher and the seller's cost is lower. The work of \citet{mcafee2008gains} shows that the fixed-price mechanism achieves a $\frac{1}{2}$-approximation of the first best (gains from trade) when the median of the buyer's value (which is a random variable) is at least the median of the seller's cost. \citet{DBLP:journals/geb/BlumrosenD21} show that the fixed-price mechanism can achieve a $\frac{1}{e}$-approximation of the first best when the buyer's value distribution has monotone hazard rate. However, for general instances, \citet{DBLP:conf/wine/BlumrosenM16} show that the fixed-price mechanism cannot always achieve a constant-approximation. Instead, the random-offerer mechanism, which uniformly randomly delegates the pricing power to the seller or the buyer, can provide such a guarantee \citep{DBLP:conf/stoc/DengMSW22}. The state-of-the-art approximation ratio of $\frac{1}{3.15}$ to the first-best is obtained by the delegation mechanisms  \citep{DBLP:conf/wine/Fei22}, and it is known that no mechanism can guarantee a ratio better than $\frac{2}{e}$ \citep{DBLP:conf/wine/BlumrosenM16}.

In addition to bilateral trade with independent valuations, efficiency approximations in bilateral trade have also been studied in settings with double auctions and matching markets (e.g., \citep{DBLP:conf/sigecom/BrustleCWZ17,DBLP:conf/sigecom/Babaioff0GZ18,DBLP:conf/soda/CaiGMZ21,DBLP:conf/stoc/0001LMZ24}), with correlated valuations~\citep{MotocCorrelated2021,DBLP:conf/stoc/DobzinskiS24}, and with a broker~\citep{DBLP:conf/soda/HajiaghayiHPS25}.
In addition to the objective of maximizing gains from trade, the different perspective of welfare maximization has also been extensively studied. \citet{DBLP:journals/geb/BlumrosenD21} provide a $(1 - \frac{1}{e})$-approximation to the first-best welfare. \citet{DBLP:conf/soda/KangPV22} improve the approximation ratio to $(1 - \frac{1}{e} + 10^{-4})$; and recently in a pair of independent works, \citet{DBLP:conf/stoc/CaiW23} and \citet{DBLP:conf/stoc/LiuRW23} further improve the ratio to $0.72$ with an upper bound of $0.7381$. Sample-based bilateral trade has received increasing attention recently for welfare optimization using the fixed-price mechanism. \citet{DBLP:conf/stoc/DuttingFLLR21} demonstrate a $\frac{1}{2}$-approximation by using a single sample from the seller's cost distribution as the posted price, and \citet{DBLP:conf/soda/KangPV22} show a $\frac{3}{4}$-approximation in a symmetric setting. \citet{DBLP:conf/stoc/CaiW23} propose a new family of sample-based fixed-price mechanisms and provide a complete characterization of their approximation ratios for any fixed number of samples. In contrast, our work focuses on developing sampled-based mechanisms for approximating the first-best gains from trade, which is a mathematically more challenging objective to provide approximations guarantees for.

\paragraph{Pricing and Auctions with Sample Access.} Our work also sits in the area of sample-based mechanism design with applications in pricing and auctions. Existing literature on sample-based mechanism design mostly focuses on two regimes: (1) determine the number of samples required to design a $(1-\varepsilon)$-optimal mechanism~\citep{cole2014sample,mohri2014learning,morgenstern2015pseudo,morgenstern2016learning,syrgkanis2017sample,cai2017learning,guo2019settling,brustle2020multi,guo2020sample,gonczarowski2021sample,balcan2023generalization}; and (2) determine the best approximation ratio a sample-based mechanism can obtain using a small fixed number of samples~\citep{dhangwatnotai2010revenue,babaioff2018two,daskalakis2020more,10.1145/3465456.3467572,DBLP:conf/stoc/DuttingFLLR21,DBLP:conf/stoc/CaiW23}. In comparison to our work, all these works consider a sampling process fully controlled by the mechanism designer, while our work assumes that the sampling process is delegated to either the seller or the buyer, and they can employ a broad range of sampling strategies.

\section{Preliminaries}
\subsection{Bilateral Trade}
We study the bilateral trade problem in which a seller and a buyer trade an indivisible item. The seller, denoted by \(\seller\), initially holds the item and incurs a private cost \(\valS\), drawn from a distribution \(\distS\), if the item is sold. Similarly, the buyer assigns a private value \(\valB\) to the item, which is drawn from a distribution \(\distB\). Both \(\valS\) and \(\valB\) are private information known only to the seller and the buyer, respectively.

We are interested in analyzing the \emph{gains from trade} achieved by a mechanism. Let \(x(\valS, \valB)\) be the probability of trade between the seller with cost \(\valS\) and the buyer with value \(\valB\) in a given mechanism. We define gains from trade as the expected improvement in total utility resulting from the trade, that is,
\[
    \gft = \E_{\valS \sim \distS, \, \valB \sim \distB}\left[(\valB - \valS) \cdot x(\valS, \valB)\right].
\]
Ideally, trade occurs when \(\valB > \valS\) and does not occur when \(\valS > \valB\). Therefore, the maximum possible gains from trade, called the \emph{first best}, can be expressed as
\[
\firstbest = \E_{\valS \sim \distS, \, \valB \sim \distB}\InBrackets{(\valB - \valS) \cdot \indic{\valB > \valS}}.
\]

In this article, we compare the gains from trade of the mechanism we propose with those of the \emph{seller-pricing mechanism} and the \emph{buyer-pricing mechanism}.

In the seller-pricing mechanism, the seller posts a take-it-or-leave-it price \(r_{\valS}\) based on her private cost \(\valS\) in order to maximize her utility. The buyer observes the price and chooses to purchase the item if and only if \(\valB \geq r_{\valS}\). The gains from trade in the seller-pricing mechanism, denoted by \(\sellerp\), can be expressed as
\[
\sellerp = \E_{\valS \sim \distS, \, \valB \sim \distB}\left[(\valB - \valS) \cdot \indic{\valB \geq r_{\valS}}\right], \quad \text{where} \quad r_{\valS} \in \argmax_p (p - \valS)\cdot \Pr_{\valB \sim \distB}[\valB \geq p].
\]

The buyer-pricing mechanism is symmetric to the seller-pricing mechanism. The key difference is that now it is the buyer who posts a price \(r'_{\valB}\) based on his private value \(\valB\), and the seller then decides whether to accept that price. Similarly, the gains from trade in the buyer-pricing mechanism, denoted by \(\buyerp\), can be expressed as
\[
\buyerp = \E_{\valS \sim \distS, \, \valB \sim \distB}\left[(\valB - \valS) \cdot \indic{\valS \leq r'_{\valB}}\right], \quad \text{where} \quad r'_{\valB} \in \argmax_p (\valB - p)\cdot \Pr_{\valS \sim \distS}[\valS \leq p].
\]
It is shown in the work of \citep{DBLP:conf/stoc/DengMSW22,DBLP:conf/wine/Fei22} that either the seller-pricing mechanism or the buyer-pricing mechanism can approximate the first best, the optimal gains from trade. Specifically, we have the following guarantee.
\[
3.15 \cdot \max(\sellerp, \buyerp) \geq \firstbest.
\]

\subsection{Agents' Behavior with Sample Access}
In this article, we assume that the entity that sets the price only has sample access to the prior distribution of the other party. For example, suppose that we are running the seller-pricing mechanism and the seller has the power to set a price. In this case, the seller knows her private cost \(\valS\) but lacks the knowledge of the prior distribution $\distB$ of the buyer, and therefore she chooses her price based on samples from $\distB$ instead. The seller can repeatedly draw i.i.d.\@ samples from the buyer's distribution $\distB$ and can adaptively decide when to stop sampling and select a price. We make the following two assumptions about the seller's behavior.
\begin{enumerate}
    \item She observes at least one sample above her private cost before setting her price. (An alternative assumption is that she sets the price at her cost if she fails to observe any sample above her cost.)
    \item When she sets the price based on the samples she has observed so far, the empirical profit of the chosen price must be at least a $c$-fraction of the empirical profit of any price. Here, the empirical profit refers to the profit evaluated on the collection of samples of $\valB \sim \distB$ that the seller has observed.
\end{enumerate}
\cref{def:SellerEO} below captures these assumptions formally.

\begin{definition}[Seller \(c\)-EO]
\label{def:SellerEO}
Suppose that the seller with private cost \(\valS\) is trying to decide her price by adaptively observing samples from the value distribution $\distB$ of the buyer. Let \(\valB^{(1)}, \valB^{(2)}, \valB^{(3)}, \ldots \overset{\mathrm{i.i.d.}}{\sim} \distB\) be the stream of samples drawn independently from \(\distB\). Assume that the seller determines the price \(r_s\) after drawing \(k\) samples, where \(k\) may be adaptively chosen based on the seller's strategy and the sample stream. For any constant $c \in (0,1]$, the pricing strategy is said to be \emph{\(c\)-EO} (``EO'' stands for ``empirically optimal'') for the seller if and only if the following two conditions hold.
\begin{itemize}
    \item \(\exists i\in[k] \text{ s.t. } \valB^{(i)} > \valS\): At least one sample has a value greater than \(\valS\).\footnote{In the case in which no sample has value larger than \(\valS\), the seller effectively learns no information from the buyer's distribution for setting the price. Therefore, we omit this case from consideration.}
    \item \((r_s - \valS) \cdot \frac{\left|\left\{i\in [k] \,\big|\, \valB^{(i)} \geq r_s\right\}\right|}{k} \geq  c \cdot\max_{p\in \mathbb{R}} (p - \valS) \cdot  \frac{\left|\left\{i\in [k] \,\big|\, \valB^{(i)} \geq p\right\}\right|}{k}  \): The empirical profit of \(r_s\) is at least a \(c\)-fraction of the optimal empirical profit.
\end{itemize}
\end{definition}

Consider the following variant of the seller-pricing mechanism in which the seller sets the take-it-or-leave-it price to be \(r_s\) chosen by any \(c\)-EO pricing strategy and the buyer decides whether to purchase based on the condition that \(\valB\geq r_s\). The gains from trade generated by this mechanism are represented as \(c\text{-}\sellersample\), defined as
\[
    c\text{-}\sellersample = \E_{\valB\sim \distB, \valS\sim\distS}\InBrackets{\InParentheses{\valB - \valS}\indic{\valB\geq r_s}}.
\]

Alternatively, when the buyer has the power to set the price, his behavior is symmetric to that of the seller. The buyer can draw samples from the seller's distribution $\distS$ and adaptively use the empirical distribution to determine the price. We similarly require that the buyer observes at least one sample with a cost no greater than his private value and selects a price such that the empirical utility of the buyer is at least a \(c\)-fraction of the optimal utility.

\begin{definition}[Buyer \(c\)-EO]
\label{def:BuyerEO}
Suppose that a buyer with private value \(\valB\) has access to a stream of samples \(\valS^{(1)}, \valS^{(2)}, \valS^{(3)}, \ldots \overset{\mathrm{i.i.d.}}{\sim} \distS\) and uses them to select a price \(r_b\). Assume that the buyer sets the price \(r_b\) after observing \(k\) samples, where \(k\) is adaptively determined based on the buyer's strategy and the observed sample stream sampled from the distribution of the seller. For any constant $c \in (0,1]$, the pricing strategy is said to be \(c\)-EO for the buyer if and only if the following two conditions are satisfied:
\begin{itemize}
    \item \(\exists i\in [k] \text{ s.t. } \valS^{(i)} < \valB\): At least one sample has a cost less than \(\valB\).\footnote{Similarly, we omit the case in which no sample can be lower than \(\valB\).}
    \item \((\valB - r_b)\cdot \frac{\left|\left\{i\in [k] \,\big|\, \valS^{(i)}\leq r_b\right\}\right|}{k} \geq c \cdot \max_{p\in \mathbb{R}} (\valB - p)\cdot \frac{\left|\left\{i\in [k] \,\big|\, \valS^{(i)}\leq p\right\}\right|}{k}\): The empirical utility of \(r\) is at least a \(c\)-fraction of the optimal empirical utility.
\end{itemize}
\end{definition}

Similarly, we consider the variant of the buyer-pricing mechanism in which the buyer sets the price to be \(r_b\) decided by the pricing strategy. We use \(c\text{-}\buyersample\) to represent the gains from trade induced by such mechanism.
\[
    c\text{-}\buyersample = \E_{\valB\sim \distB, \valS\sim\distS}\InBrackets{\InParentheses{\valB - \valS}\indic{\valS\leq r_b}}.
\]

\section{Main Result}

We now present our main theorem, demonstrating that when both the seller and the buyer employ \(c\)-EO pricing strategies, at least one of the two pricing mechanisms provides an \(O(c)\)-approximation to the optimal gains from trade. The main result is formally stated as follows.
\begin{theorem}
\label{thm:main}
Fix the distributions $\distS$ and $\distB$ of the seller and the buyer. Suppose that the seller and the buyer determine their prices according to arbitrary \(c\)-EO pricing strategies for a constant $c \in (0, 1]$. Let \(c\text{-}\sellersample\) and \(c\text{-}\buyersample\) respectively represent the gains from trade induced by the seller's and buyer's \(c\)-EO pricing strategies. It holds that
\[
\frac{25.2}{c} \cdot\max\InParentheses{c\text{-}\sellersample, c\text{-}\buyersample} \geq \firstbest.
\]
As an important special case, when the seller and the buyer use $1$-EO (exactly empirically optimal) strategies, then
\[
25.2 \cdot\max\InParentheses{1\text{-}\sellersample, 1\text{-}\buyersample} \geq \firstbest.
\]
\end{theorem}

Our proof proceeds as follows. Observe that, for the seller-pricing mechanism, the gains from trade
\[
\E\InBrackets{(\valB - \valS) \cdot \indic{\valB \geq r_{\valS}}}
\]
can be decomposed into two components:
\begin{itemize}
    \item the profit of the seller, \( \E[(r_{\valS} - \valS) \cdot \indic{\valB \geq r_{\valS}}] \), and
    \item the utility of the buyer, \( \E[(\valB - r_{\valS}) \cdot \indic{\valB \geq r_{\valS}}] \).
\end{itemize}
We denote these two components as \(\srev\) and \(\suti\) respectively, and define them as follows.
\begin{align*}
\begin{split}
    \srev &= \E_{\valS\sim \distS, \valB\sim \distB}\left[(r_{\valS} - \valS) \cdot \indic{\valB \geq r_{\valS}}\right],\\
    \suti &= \E_{\valS\sim \distS, \valB\sim \distB}\left[(\valB - r_{\valS}) \cdot \indic{\valB \geq r_{\valS}}\right].
\end{split}
\end{align*}
Similarly, we define analogous terms for the buyer-pricing mechanism as follows.
\begin{align*}
\begin{split}
    \brev &= \E_{\valS\sim \distS, \valB\sim \distB}\left[(\valB - r_{\valB}) \cdot \indic{\valS \leq r_{\valB}}\right],\\
    \buti &= \E_{\valS\sim \distS, \valB\sim \distB}\left[(r_{\valB} - {\valS}) \cdot \indic{\valS \leq r_{\valB}}\right].
\end{split}
\end{align*}

Existing results \cite{DBLP:conf/stoc/DengMSW22,DBLP:conf/wine/Fei22} on the efficiency of the seller-pricing and buyer-pricing mechanisms provide a stronger claim: the seller's profit from the seller-pricing mechanism, combined with the buyer's profit from the buyer-pricing mechanism, suffices to approximate the optimal gains from trade.
\begin{lemma}[\cite{DBLP:conf/wine/Fei22}]
\label{lem:pricing_efficency}
For any $\distS$ and $\distB$, we have
\[
3.15 \cdot \max(\srev, \brev) \geq \firstbest.
\]
\end{lemma}

When the seller or buyer lacks full knowledge of the distributions and instead employs a \(c\)-EO pricing strategy to determine 
their prices in the seller-pricing or buyer-pricing mechanisms, we aim to demonstrate that the gains from trade generated 
by these mechanisms are sufficient to cover the \emph{profits} achieved in the original seller-pricing or buyer-pricing mechanisms.

\begin{theorem}
\label{thm:sample_approximate_revenue}
Assume that the seller and the buyer are running any \(c\)-EO pricing strategies to determine their prices for some arbitrary constant $c \in (0,1]$. Let \(c\text{-}\sellersample\) and \(c\text{-}\buyersample\) respectively denote the gains from trade resulting from these \(c\)-EO strategies. Then we have
\[
c\text{-}\sellersample \geq \frac{c}{8}\cdot \srev \quad \text{and} \quad c\text{-}\buyersample\geq \frac{c}{8}\cdot \brev.
\]
\end{theorem}

\cref{thm:sample_approximate_revenue} serves as our main technical result, and its proof is deferred to \Cref{sec:walk}. Combining \Cref{lem:pricing_efficency} and \Cref{thm:sample_approximate_revenue}, we can directly establish our main theorem.

\begin{proof}[Proof of \Cref{thm:main}]
First, note that \Cref{thm:sample_approximate_revenue} establishes the following bounds.
\begin{align}\label{eq:sample_approximate_revenue}
c\text{-}\sellersample \geq \frac{c}{8}\cdot \srev \quad \text{and} \quad c\text{-}\buyersample\geq \frac{c}{8}\cdot \brev.
\end{align}

Additionally, \Cref{lem:pricing_efficency} states
\begin{align}\label{eq:pricing_efficency}
    3.15 \cdot \max(\srev, \brev) \geq \firstbest.
\end{align}

Combining \cref{eq:sample_approximate_revenue,eq:pricing_efficency}, we derive
\begin{align*}
    \frac{25.2}{c} \cdot \max\left(c\text{-}\sellersample, c\text{-}\buyersample\right) \geq \firstbest.
\end{align*}
This completes the proof.
\end{proof}
 
\section{Sample-Based Welfare Approximates Optimal Revenue}
\label{sec:walk}

In this section, we complete our proof of \Cref{thm:sample_approximate_revenue} and aim to show that
\[
c\text{-}\sellersample \geq \frac{c}{8} \cdot \srev.
\] 
We only focus on the case in which the seller has the pricing power, since the proof for the buyer-pricing case is completely symmetric.

To establish this approximation of $c\text{-}\sellersample$ to $\srev$, we begin by fixing the seller's cost, \(\valS\), and normalizing it to \(0\). Note that this normalization reduces our problem to the one in \cref{def:welfare_apx_revenue} below.

\begin{definition}[Sample-Based Welfare Approximates Optimal Revenue]
\label{def:welfare_apx_revenue}
Consider an unknown distribution \(\dist\) supported on \(\mathbb{R}\), with the assumption that the optimal revenue derived from this distribution is positive. Specifically, we have
\[
\max_{p \in \mathbb{R}} p \cdot \Pr_{\val \sim \dist}\left[\val \geq p\right] > 0.
\]
The seller is allowed to sample repeatedly from the distribution $F$. She may choose to stop at any time and post an arbitrary price \(r\),  subject to the two requirements of \(c\)-EO.
\begin{itemize}
    \item The empirical revenue achieved at price \(r\) must be at least a \(c\)-fraction of the optimal empirical revenue.
    \item The optimal empirical revenue at the stopping point must be positive. 
\end{itemize}

Fix the distribution \(\dist\). We define \(\wel(p)\) and \(\rev(p)\) as the welfare and the revenue induced by posting a price \(p\), respectively. Specifically, we have
\begin{align*}
    \wel\InParentheses{p} = \E_{\val \sim \dist}\InBrackets{\val\cdot \indic{\val \geq p}} \quad \text{and} \quad \rev\InParentheses{p} = p \cdot \Pr_{\val\sim \dist}\InBrackets{\val \geq p}.
\end{align*}

Our objective is to show that the welfare at the price \(r\) chosen by the seller is at least a constant fraction of the optimal revenue. Formally, we aim to prove
\[
\wel\InParentheses{r} \geq \Omega(c) \cdot \max_{p\in \mathbb{R}} \rev\InParentheses{p}.
\]
\end{definition}

\cref{lem:general_prob} below asserts that, with probability at least \(1 - 2\delta/c\), the seller will set a price such that the welfare is no less than \(\delta\) times the optimal revenue, for any \(\delta \in (0, 1)\) and any strategy employed by the seller.

\begin{theorem}
\label{lem:general_prob}
For any $\delta > 0$ and \(c\in [2\delta, 1]\), in all instances and for all \(c\)-EO strategies of the seller, the probability that she picks a price with welfare less than $\delta$ times the optimal revenue is at most $2\delta/c$.
\end{theorem}

\cref{thm:lower_bound} below states that \Cref{lem:general_prob} is tight up to a constant of \(2\). The proof of \Cref{thm:lower_bound} is deferred to \Cref{appendix:proof_of_lowerbound}.

\begin{proposition}
\label{thm:lower_bound}
For any \(\delta > 0\) and \(c \in [\delta, 1]\), there exists an instance and a seller's \(c\)-EO strategy such that the probability of selecting a price with welfare less than \(\delta\) times the optimal revenue is exactly \(\delta/c\).
\end{proposition}

Before proving \Cref{lem:general_prob}, we first show how it can imply \Cref{thm:sample_approximate_revenue}.

\begin{proof}[Proof of \Cref{thm:sample_approximate_revenue} assuming \Cref{lem:general_prob}]
Fix the seller's cost \(\valS\). Define \(\dist\) as the distribution obtained by shifting \(\dist_B\) to the left by \(\valS\). Since the seller decides the price according to some \(c\)-EO pricing strategy in the old setting, it is equivalent to our new setting in which the seller observes at least one positive sample from the distribution \(\dist\), and then selects a price so that the empirical revenue is at least \(c\) fraction of the optimal empirical revenue. As the seller observes at least one positive sample, the optimal empirical revenue is positive. Therefore, by \Cref{lem:general_prob}, with a probability of at least \(1 - 2\delta/c\), the generated welfare of the selected price \(r_s\) is at least \(\delta\) fraction of the optimal revenue. This implies
\begin{align*}
	\E_{\valB \sim \distB} \left[(\valB - \valS) \cdot \indic{\valB \geq r_s}\right] \geq \left(1 - \frac{2\delta}{c}\right) \cdot \delta \cdot \max_{p\in \mathbb{R}} \E_{\valB \sim \distB} \left[(p - \valS) \cdot \indic{\valB \geq p}\right].
\end{align*}

By setting \(\delta = c / 4\), we obtain
\begin{align*}
	\E_{\valB \sim \distB} \left[(\valB - \valS) \cdot \indic{\valB \geq r_s}\right] \geq \frac{c}{8} \cdot \max_{p\in \mathbb{R}} \E_{\valB \sim \distB} \left[(p - \valS) \cdot \indic{\valB \geq p}\right].
\end{align*}

Now taking the expectation with respect to \(\valS\sim \distS\), it follows that
\begin{align*}
		\E_{\valS\sim \distS, \valB \sim \distB} \left[(\valB - \valS) \cdot \indic{\valB \geq r_s}\right] \geq \frac{c}{8} \cdot\E_{\valS\sim \distS}\left[  \max_{p\in \mathbb{R}} \E_{\valB \sim \distB} \left[(p - \valS) \cdot \indic{\valB \geq p}\right]\right].
\end{align*}	
Note that the left-hand side corresponds to the definition of \(c\text{-}\sellersample\), while the right-hand side matches the definition of \(\srev\). Therefore, we conclude that
\begin{align*}
	c\text{-}\sellersample \geq \frac{c}{8}\cdot \srev,
\end{align*}
which completes the proof.
\end{proof}

It remains to prove \Cref{lem:general_prob}. In particular, we aim to characterize the worst-case combination of the buyer's distribution and the seller's sampling strategy that maximizes the seller's probability of selecting a price that yields a welfare less than a \(\delta\)-fraction of the optimal revenue in hindsight. Given the buyer's distribution, the worst-case sampling strategy for the seller is straightforward: the seller keeps sampling until the seller can post a price such that the price can achieve an empirical revenue of at least a \(c\)-fraction of the optimal empirical revenue and generate welfare less than a \(\delta\)-fraction of the optimal revenue in hindsight. Thus, the main technical challenge lies in characterizing the worst-case distribution.

Our proof consists of three steps:
\begin{enumerate}
    \item As our first (and easy) step, we demonstrate that it is without loss of generality to consider simple forms of distributions. Firstly, it is sufficient to focus on distributions with discrete supports. Furthermore, we show that, without loss of generality, the distribution has an optimal revenue of \(1\) achieved at a price of \(1\). This implies that in the worst-case distributions, no probability mass exists below \(1\). The formal proof of these results is provided in \Cref{subsec:worst_case_shape}.
    
    \item As the second step, we consider the case where the support of the distribution consists exclusively of integer multiples of \(\frac{1}{c}\). We demonstrate that, in this scenario, the worst-case distributions have binary supports. The formal proof is presented in \Cref{subsec:integer_supp}.

    \item As the final step, we apply the results from the integer case to address the general case. The formal proof is provided in \Cref{subsec:general_supp}.
\end{enumerate}

\subsection{Structure of Worst-Case Distributions}
\label{subsec:worst_case_shape}

For any distribution \(\dist\) with a strictly positive optimal revenue, it can always be assumed, without loss of generality, that \(\dist\) achieves its optimal revenue at a price of \(1\). To see this, let \(p \in \argmax \left\{p \cdot \Pr[\val \geq p]\right\}\) be an arbitrary optimal price for the distribution. It is clear that \(p\) is positive, as its revenue is positive. By rescaling the distribution by a factor of \(1/p\), we can normalize the optimal price to \(1\).

We now prove that it is also without loss of generality to assume that the distribution has no probability mass below \(1\).

\begin{lemma}
\label{lem:no_mass_under_one}
Fix any distribution \(\dist \in \Delta(\mathbb{R})\) with an optimal price of \(1\) and assume the worst-case strategy of the seller. Let $\ell$ be the probability that the seller selects a price \(r\) resulting in welfare less than a \(\delta\)-fraction of the optimal revenue. Then, there exists another distribution \(\dist'\) with no probability mass below \(1\) and a seller's strategy such that the probability of selecting a price \(r\) with welfare less than a \(\delta\)-fraction of the optimal revenue is at least \(\ell\).
\end{lemma}
\begin{proof}
    We first notice that the welfare function is non-increasing over \(\mathbb{R}_{\geq 0}\). Specifically, for any \(x \geq y \geq 0\), it holds that \(\wel(x) \geq \wel(y)\). Intuitively, this implies that under the worst-case strategy, the seller will never choose a price below \(1\). Therefore, we can remove the probability mass below \(1\) and this only increases the probability that a bad price is chosen. The formal proof is provided below.

    Define \(\dist_{\geq 1}\) as the conditional distribution of \(\dist\) given that the value is at least \(1\). Formally,
    \begin{equation*}
        \Pr_{\val \sim \dist_{\geq 1}}\InBrackets{ \val\geq x} = \left\{
        \begin{aligned}
            &1  & x \leq 1,\\
            &\frac{\Pr_{\val \sim \dist}[\val \geq x]}{\Pr_{\val \sim \dist}[\val \geq 1]}  & x \geq 1.
        \end{aligned}
        \right.
    \end{equation*}

    Now suppose that the seller employs the worst-case strategy on both distributions \(\dist\) and \(\dist_{\geq 1}\), and we couple these two random processes as follows. The seller draws a sample \(x\) from \(\dist\) and updates her empirical distribution for \(\dist\). If \(x \geq 1\), the same sample \(x\) is also used to update the empirical distribution of \(\dist_{\geq 1}\). Since the seller is following the worst-case strategy, she will stop the sampling process and select a price \(r\) if and only if a price \(r\) satisfies \(\wel(r) < \delta \cdot \rev(1)\) and has empirical revenue exceeding a \(c\)-fraction of the optimal revenue.

    Since the welfare function is monotone, for any \(r \leq 1\), it follows that
    \begin{align*}
        \wel\InParentheses{r} \geq \wel\InParentheses{1} \geq \rev\InParentheses{1}.
    \end{align*}
This implies that during the sampling process of the distribution \(\dist\), the seller will only stop at a price \(r > 1\) such that \(\wel(r) < \delta \rev(1)\). Let \(\emprev(x)\) and \(\emprev_{\geq 1}(x)\) denote the empirical revenue at price \(x\) based on the sampling processes of \(\dist\) and \(\dist_{\geq 1}\), respectively, at the time the seller stops sampling and selects the price \(r\) for \(\dist\).
 It is known that \(r\) has an empirical revenue of at least a \(c\)-fraction of the optimal empirical revenue, that is,
\[
\emprev\InParentheses{r} \geq c \cdot \max_{p} \, \emprev\InParentheses{p}.
\]
Similarly, let \(\tilde{\dist}(x)\) and \(\tilde{\dist}_{\geq 1}(x)\) denote the complementary cumulative distribution functions (denoting $\Pr[X \geq x]$ for a distribution $X$) of the empirical distribution at the time the seller stops. Then, the empirical revenue functions are given by:
\[
\emprev(x) = x \tilde{\dist}(x), \quad \emprev_{\geq 1}(x) = x \tilde{\dist}_{\geq 1}(x).
\]

Given our sampling process, it is clear that the samples with value at least \(1\) are exactly the same for \(\dist\) and \(\dist_{\geq 1}\). Therefore, the following holds for any \(a, b \geq 1\).
\[
\frac{\Tilde{F}(a)}{\Tilde{F}(b)} = \frac{\Tilde{F}_{\geq 1}(a)}{\Tilde{F}_{\geq 1}(b)}.
\]
Combining with the inequality above and the definition of the empirical revenue function, it follows that
\[
\emprev_{\geq 1}\InParentheses{r} \geq c \cdot \max_{p \geq 1} \, \emprev_{\geq 1}\InParentheses{p}.
\]
Note that in the distribution \(\dist_{\geq 1}\), the welfare induced by \(r\) is also less than a \(\delta\)-fraction of the optimal revenue. Thus, the seller will also terminate and output price \(r\) in this case. This completes our proof.
\end{proof}

Until now, we have shown that the distribution has no probability mass below \(1\) and achieves its maximum revenue at a price of \(1\). Since the optimal revenue is normalized to \(1\) and there is no probability mass below \(1\), the empirical optimal revenue is always at least \(1\). In the remainder of this proof, we establish a stronger result: we now assume that the seller can select a price with welfare below \(\delta\) whenever its empirical revenue is at least \(c\), rather than when it is at least a \(c\)-fraction of the optimal empirical revenue, which is always weakly higher. This assumption strengthens the worst-case strategy by making it easier for the seller to choose an undesired price. Under this stronger assumption, we demonstrate that it is sufficient to restrict the analysis to discrete distributions.

\begin{lemma}
\label{lem:discretize}
Consider any (potentially non-discrete) distribution \(\dist\) with no probability mass below \(1\) and an optimal revenue of \(1\) achieved at a price of \(1\). For any \(\delta > 0\), assume that, with probability \(\ell\), there exists a price with welfare not exceeding \(\delta\) such that its empirical revenue is at least \(c\) after a certain number of samples. For any \(\varepsilon > 0\), there exists a corresponding discrete distribution \(\dist'\) such that a price also exists with welfare at most \((\delta + \varepsilon)\) and empirical revenue of at least \(c\), after sampling a certain number of times, with probability at least \(\ell\).
\end{lemma}

The proof of \Cref{lem:discretize} is deferred to \Cref{appendix:proof_of_discretize}. Leaving \(\varepsilon \to 0\), it follows that we can restrict our consideration to discrete distributions without loss of generality.

\subsection{Integer Support}
\label{subsec:integer_supp}

As described in \Cref{subsec:worst_case_shape}, we assume that the distribution \(\dist\) is structured as follows: it assigns a value \(\val_i\) with probability \(\prob_i\) for \(i = 0, 1, \ldots, k\), where \(k\) may potentially be infinite. We further assume that \(\val_i < \val_{i+1}\) and \(\val_0 = 1\). Additionally, the optimal revenue is achieved at \(1\). Formally, for any \(i \in \{0, 1, \ldots, k\}\), the following inequality holds:
\[
\val_i \cdot \left(\sum_{j = i}^{k} \prob_j\right) \leq 1.
\]

In this section, we focus on the case where all support values are integer multiples of \(1/c\). Specifically, we assume that each \(\val_i/c\) is integer for all \(i \in \{1, 2, 3, \ldots, k\}\).

We begin by analyzing the case of binary support. The following lemma provides a complete characterization of this scenario.

\begin{lemma}
\label{lem:binary_supp}
Consider a binary support distribution where \(k = 1\) and \(\frac{\val_1}{c}\) is an integer. Assume that \(\frac{\prob_1\val_1}{c} \leq 1\). Let \(\emprev_{t}\InParentheses{x}\) denote the empirical revenue of price \(x\) after \(t\) samples. The probability that there exists some \(t \in \mathbb{N}_{+}\) such that \(\emprev_{t}\InParentheses{\val_1} \geq c\) is precisely \( \frac{\prob_1\val_1}{c}\).
\end{lemma}
\begin{proof}

Consider the following random walk over \(\mathbb{Z}\). Let \(X_0 = a \in \mathbb{R}\), and define \(X_i\) as follows:
\[
X_i = 
\begin{cases} 
    X_{i - 1} - 1 & \text{with probability } \prob_0, \\
    X_{i - 1} + \frac{\val_1}{c} - 1 & \text{with probability } \prob_1.
\end{cases}
\]

We first argue that the probability of there existing some \(t \in \mathbb{N}_{+}\) such that the empirical revenue of \(\val_1\) after \(t\) samples is at least \(c\) is identical to the probability that there exists some \(t > 0\) for which \(X_t \geq 0\), given the initialization \(X_0 = 0\). This equivalence is established by coupling the sampling process with the random walk defined above.  

For any \(i \geq 1\), if \(X_i = X_{i - 1} - 1\), we assume that the corresponding sample is \(x_i = 1\); otherwise, it is \(x_i = \val_1\). Clearly, the sequence \(\{x_i\}\) corresponds to a stream of i.i.d.\@ samples from the binary distribution. 

Now, suppose \(X_t \geq 0\) for some \(t > 0\). Let \(\ell\) denote the number of steps with \(X_i = X_{i-1} + \frac{\val_1}{c} -1\) in the first \(t\) steps of the random walk. Notice that \(\ell\) also represents the number of occurrences of \(\val_1\) in the first \(t\) samples \(x_1, x_2, \dots, x_t\). We observe that  
\[
X_t \geq 0 \Leftrightarrow \ell \cdot \frac{\val_1}{c} \geq t \Leftrightarrow \val_1 \cdot \frac{\ell}{t} \geq c \Leftrightarrow \emprev_t(\val_1) \geq c.
\]
Thus, we establish that \(X_t \geq 0\) if and only if \(\emprev_t(\val_1) \geq c\). In the following, we aim to characterize the probability that \(X_t \geq 0\) for some \(t > 0\).

Let \(f(a)\) denote the probability that, starting from \(X_0 = a\), there exists some \(t > 0\) such that \(X_t \geq 0\). Similarly, let \(g(a)\) represent the probability that, starting from \(X_0 = a\), there exists some \(t > 0\) such that \(X_t > 0\). 

Given the recursion of the random walk, the following fact about $g(a)$ holds.
\begin{align*}
    g(a) = 
    \begin{cases}
        1 & a \geq 1,\\
        \prob_1 \cdot g\big(a + \frac{\val_1}{c} - 1\big)  + (1 - \prob_1) \cdot g(a - 1) & a < 1.
    \end{cases}
\end{align*}

Now, define \(h(a) = g(1 - a) - g(-a)\). The definition of \(h(a)\) gives the following fact.
\begin{align*}
    h(a) = 
    \begin{cases}
        0 & a < 0,\\
        \frac{\prob_1}{1 - \prob_1} \cdot \displaystyle\sum_{i=1}^{\frac{\val_1}{c} - 1} h(a-i) & a > 0.
    \end{cases}
\end{align*}
To see why \(h(a)\) takes this form when \(a > 0\), notice that 
\begin{align*}
h(a) &= g(1 - a) - g(-a) \\
    & = \left(p_1\cdot g\InParentheses{\frac{\val_1}{c}-a} + (1 - p_1)\cdot g( - a)\right) - g(-a)\\
    & = p_1 \left(g\InParentheses{\frac{\val_1}{c}-a} - g(-a)\right)\\
    & = p_1 \sum_{i = 0}^{\frac{\val_1}{c} - 1} g\InParentheses{i - a + 1} - g(i-a)\\
    & = p_1 \sum_{i = 0}^{\frac{\val_1}{c} - 1} h(a - i).
\end{align*}

By moving \(h(a)\) from the right hand side to the left hand side, it follows that 
\[
h(a) = \frac{p_1}{1 - p_1} \sum_{i = 1}^{\frac{\val_1}{c} - 1} h(a - i).
\]

It is straightforward that \(\lim_{a\rightarrow -\infty} g(a) = 0\). Therefore, This means that 
\begin{equation}
\label{eq:solve_ga_1}    
\sum_{a=0}^{\infty} h(a) = g(1) - \lim_{a\rightarrow-\infty} g(a) = 1.
\end{equation}

From the recursion of \(h(a)\), we derive
\begin{align}\label{eq:solve_ga_2}
\begin{split}
    \sum_{a=0}^{\infty} h(a) & = h(0) + \sum_{a = 1}^{\infty} h(a)\\
        & = h(0) + \frac{\prob_1}{1 - \prob_1} \cdot \InParentheses{\frac{\val_1}{c} - 1} \cdot \sum_{a = 0}^{\infty} h(a).
\end{split}
\end{align}

Combining \cref{eq:solve_ga_1,eq:solve_ga_2}, it follows that 
\[
h(0) = 1 - \frac{\prob_1}{1 - \prob_1} \cdot \InParentheses{\frac{\val_1}{c}-1}.
\] 
Therefore,
\[
g(0) = g(1) - h(0) = \frac{\prob_1}{1 - \prob_1} \cdot \left(\frac{\val_1}{c} - 1\right).
\]

Finally, observe that \(f(0) = \prob_1 + (1 - \prob_1) \cdot g(0)\). This follows because, with probability \(\prob_1\), we directly have \(X_1 > 0\); otherwise, with probability \(1 - \prob_1\), the random walk starts at \(-1\) and continues until it reaches a non-negative point, which is precisely described by \(g(0)\). Substituting, we have
\[
f(0) = \prob_1 + (1 - \prob_1) \cdot g(0) = \frac{\prob_1 \val_1}{c}.
\]
This completes the proof.
\end{proof}

In the following, we show that binary support is the worst-case distribution. The following lemma states that the probability that a price with welfare less than \(\delta\) is chosen is at most \(\delta/c\).
\begin{lemma}
\label{lem:integer_multi_supp}
Consider an arbitrary distribution where \(\frac{\val_i}{c}\) is an integer for all \(i\in \{1, 2, \ldots, k\}\). Let \(\emprev_{t}\InParentheses{x}\) denote the empirical revenue of price \(x\) after \(t\) samples. For any \(i\in \{1, 2, \ldots, k\}\), the probability that there exists some \(t \in \mathbb{N}_{+}\) and \(j \geq i\) such that \(\emprev_{t}\InParentheses{\val_j} \geq c\) is at most \(\wel\InParentheses{\val_i} / c\).
\end{lemma}
\begin{proof}
    We begin by defining the event \(\mathcal{F}_i\): there exists some \(t \in \mathbb{N}_{+}\) such that the empirical revenue of \(\val_i\) after \(t\) samples is at least \(c\), while for all \(t \in \mathbb{N}_{+}\), the empirical revenue of \(\val_{i+1}\) remains below \(c\). Specifically, for \(i = k\), we define \(\val_{k+1}\) and \(\prob_{k+1}\) to be \(0\). We aim to show that 
    \[
    \Pr[\mathcal{F}_i] \leq \frac{\prob_i \val_i}{c}.
    \]

To establish this, we define two events, \(\mathcal{G}_i\) and \(\mathcal{H}_i\). The event \(\mathcal{G}_i\) occurs if there exists some \(t \in \mathbb{N}_{+}\) such that the empirical revenue of \(\val_i\) after \(t\) samples is at least \(c\). Similarly, the event \(\mathcal{H}_i\) occurs if there exists some \(t \in \mathbb{N}_{+}\) such that the empirical revenue of \(\val_i\) after \(t\) samples is at least \(c\), while for some (not necessarily the same) \(t' \in \mathbb{N}_{+}\), the empirical revenue of \(\val_{i+1}\) after \(t'\) samples is also at least \(c\).

From \Cref{lem:binary_supp}, it directly follows that 
\[
\Pr\InBrackets{\mathcal{G}_i} = \frac{\val_i \cdot \sum_{j=i}^k \prob_j}{c},
\]
because if we focus solely on the empirical revenue of \(\val_i\), it is equivalent to a binary support case where the value \(\val_i\) has a probability of \(\sum_{j=i}^k \prob_j\), and the value is \(1\) otherwise.

We now provide a lower bound for \(\mathcal{H}_i\). To demonstrate this, we consider the following coupled sampling process. Let \(x_t\) denote the \(t\)-th sample drawn from the current distribution. Additionally, define \(x_t'\) as follows: \(x_t' = \val_i\) if \(x_t \geq \val_{i+1}\), and \(x_t' = 1\) otherwise. Similarly, we define \(\emprev_t'(r)\) as the empirical revenue at price \(r\) based on the first \(t\) samples \(\left\{x_1', x_2', \dots, x_t'\right\}\). Formally,  
\[
\emprev_t'(r) = r \cdot \frac{\left|\left\{i\in [t] \mid x_i' \geq r\right\}\right|}{t}.
\]

Now assume that there exists some \(t\in \mathbb{N}_{+}\) such that 
\[
\emprev_t'(\val_i) \geq c.
\]
By the definition of the sampling process, we have
\[
\left\{i\in [t]\given x_i \geq \val_{i+1}\right\} = \left\{i\in [t]\given x_i' \geq \val_{i}\right\}
\]
Therefore, it follows that
\begin{align*}
\emprev_t\InParentheses{\val_i} &= \val_i \cdot \frac{\left|\left\{i\in [t]\given x_i \geq \val_{i}\right\}\right|}{t} \geq  \val_i \cdot \frac{\left|\left\{i\in [t]\given x_i \geq \val_{i+1}\right\}\right|}{t} = \val_i \cdot \frac{\left|\left\{i\in [t]\given x_i' \geq \val_{i}\right\}\right|}{t} \\
& =  \emprev_t'\InParentheses{\val_i} \geq c,
\end{align*}
and
\begin{align*}
   \emprev_t\InParentheses{\val_{i + 1}} &=  \val_{i + 1} \cdot \frac{\left|\left\{i\in [t]\given x_i \geq \val_{i+1}\right\}\right|}{t} \geq \val_{i} \cdot \frac{\left|\left\{i\in [t]\given x_i \geq \val_{i+1}\right\}\right|}{t} = \val_i \cdot \frac{\left|\left\{i\in [t]\given x_i' \geq \val_{i}\right\}\right|}{t}\\
   & = \emprev_t'\InParentheses{\val_i} \geq c.
\end{align*}
This implies that \(\mathcal{H}_i\) occurs whenever \(\val_i\) achieves an empirical revenue of at least \(c\) with respect to the sampling process \(\{x_i'\}_{i=1}^{\infty}\). Furthermore, observe that \(x_i'\) is a sample from a binary distribution, where the value is \(\val_i\) with probability \(\sum_{j=i+1}^{k} \prob_j\), and \(1\) otherwise. By Lemma~\ref{lem:binary_supp}, it holds that
\[
\Pr[\mathcal{H}_i] \geq \frac{\val_i \cdot \sum_{j=i+1}^{k} \prob_j}{c}.
\]

Finally, it is clear that 
\[
\Pr\InBrackets{\mathcal{F}_i} = \Pr\InBrackets{\mathcal{G}_i} - \Pr\InBrackets{\mathcal{H}_i} \leq \frac{\val_i \prob_i}{c}.
\]

For any \(i \in \{1, 2, \dots, k\}\), the probability that there exists some \(t \in \mathbb{N}_{+}\) and \(j \geq i\) such that the empirical revenue of \(\val_j\) after sampling \(t\) times is at least \(c\) is clearly upper bounded by
\[
\sum_{j=i}^{k} \Pr\left[\mathcal{F}_j\right] \leq \sum_{j=i}^{k} \frac{\prob_j \val_j}{c} = \frac{\wel(\val_i)}{c}.
\]
This completes the proof.
\end{proof}

\subsection{General Support}
\label{subsec:general_supp}

In this section, we complete the proof of \Cref{lem:general_prob} for general supports. As in \Cref{subsec:integer_supp}, we consider a distribution where each value \(\val_i\) occurs with probability \(\prob_i\) for \(i \in \{0, 1, 2, \dots, k\}\), where \(k\) may be infinite.
 Additionally, we assume that \(\val_i < \val_{i+1}\) and that the distribution achieves an optimal revenue of \(1\) at \(\val_0 = 1\). Unlike \Cref{subsec:integer_supp}, we do not require \(\val_i/c\) to be an integer in this section.

Given any \(\delta > 0\), let \(i^{*} = \min \left\{i \mid \wel\InParentheses{\val_i} \leq \delta \right\}\). Without loss of generality, we assume that such \(i^{*}\) exists. Our goal is to show the following lemma.

\begin{lemma}
\label{lem:general_multi_supp}
Let \(\emprev_t(x)\) denote the empirical revenue of price \(x\) after \(t\) samples. For any \(i^{*} \in \{1,2,\ldots,k\}\), the probability of there existing some \(t \in \mathbb{N}_{+}\) and \(j \geq i^{*}\) such that
\[
\emprev_t(\val_j) \geq c
\]
is at most \(2 \wel\InParentheses{i^*}/c\).
\end{lemma}

Before proving \Cref{lem:general_multi_supp}, observe that it directly implies \Cref{lem:general_prob}. Note that the optimal empirical revenue after any number of samples is always at least \(1\), achieved at \(\val_0 = 1\), since all samples have a value of at least \(1\). Applying \Cref{lem:general_multi_supp} to \(i^{*}\), it follows that, with probability at most \(2\wel(\val_i)/c \leq 2\delta/c\), there will exist a price \(\val_j\) such that \(j \geq i^{*}\) achieving an empirical revenue exceeding \(c\). This implies that, with probability at most \(2\delta/c\), a price \(\val_j\) with induced welfare at most \(\delta\) will achieve an empirical revenue of at least a \(c\)-fraction of the optimal empirical revenue after a certain number of samples. Therefore, even under the seller's worst-case strategy, the probability remains at most \(2\delta/c\), which completes the proof of \Cref{lem:general_prob}.

\begin{proof}[Proof of \Cref{lem:general_multi_supp}]

Consider a new distribution \(\dist' = \left\{(\val_0', \prob_0'), (\val_1', \prob_1'), \dots, (\val_k', \prob_k')\right\}\) constructed as follows. Let \(\val_0' = 1\) and \(\prob_i' = \prob_i\) for all \(i \in \{0, 1, 2, \dots, k\}\). For each \(i \in \{1, 2, \dots, k\}\), define
\[
\val_i' = c \left\lceil \frac{\val_i}{c} \right\rceil.
\]

Given that \(\val_i > 1\) and \(c < 1\), it is evident that
\[
\val_i \leq \val_i' \leq \val_i + 1 \leq 2\val_i.
\]

For each \(i \in \{1, 2, \dots, k\}\), since \(\val_i \leq \val_i'\) and \(\prob_i = \prob_i'\), it follows directly that switching the distribution from \(\dist\) to \(\dist'\) can only increase the probability that there exist some \(t \in \mathbb{N}_{+}\) and \(j \geq i^{*}\) such that the empirical revenue of \(\val_j\) after \(t\) samples is at least \(c\). Notice that for distribution \(\dist'\), the quantity \(\val_i'/c\) is always an integer for \(i\in \{1,2,\ldots, k\}\). Therefore, applying \Cref{lem:integer_multi_supp} to \(\dist'\), it follows that such probability is at most 
\[
\frac{\sum_{j=i^{*}}^k \prob_i' \val_i'}{c} \leq \frac{2\sum_{j=i^{*}}^k \prob_i \val_i}{c} = \frac{2\wel\InParentheses{i^{*}}}{c}. \qedhere
\]
\end{proof}

\bibliographystyle{alpha}
\bibliography{ref}

\appendix

\section{Omitted Proofs}

\subsection{Proof Sketch of Revenue-Revenue and Welfare-Welfare Inapproximability}
\label{appendix:wwrr}
Here we sketch the proofs that for a $1$-EO (\cref{def:SellerEO}) seller with cost $c = 0$, we have that (1) the resulting revenue may not be a constant-factor approximation of the public-prior-case revenue, and (2) the resulting welfare may not be a constant-factor approximation of the public-prior-case welfare.
\begin{itemize}
    \item It is relatively easy to see that the resulting revenue may not approximate the optimal revenue. As an example, let the buyer's value be equal to $1$ with high probability and equal to $M^2$ with probability $\frac{1}{M}$ for a large number $M$. If the seller only gets a few samples, then she is likely to set the price at $1$ with revenue of $1$. However, the optimal revenue is $M^2 \cdot \frac{1}{M} = M$ if the seller knows the prior distribution of the buyer $\distB$.
    \item As for welfare approximation, consider a seller who always decides the price after getting exactly $k \gg M_m$ samples. Let $M_0 < M_1 < M_2 < \cdots < M_m$ be a sequence of numbers where $M_i$ is sufficiently large compared to $M_{i - 1}$, and $\ln M_0 \gg m$. Let the values be supported on
    \[
    S = \{1\} \cup [2, M_0] \cup \{M_1, M_2, \ldots, M_m\}.
    \]
    Define the buyer's value distribution $\distB$ (supported on $S$) in the following way.
    \begin{itemize}
        \item Posting a price of $p = M_i$ for $i \geq 1$ gives a revenue of $1 - \varepsilon$ for an infinitesimally small $\varepsilon > 0$.
        \item Posting a price of $p \in [2, M_0]$ gives a revenue of $\frac{1}{2}$.
    \end{itemize}
    We point out that a seller with full information of $\distB$ will set the price at $1$, resulting in a welfare of at least $\Omega(\log M_0)$. However, the $1$-EO seller with $k$ samples will post a price of at least $M_1$ with probability of about $1 - \frac{1}{m + 1}$, and if this happens, then the resulting welfare is at most $m$.
    To see why the probability is about $1 - \frac{1}{m + 1}$, note that $p \in [2, M_0]$ are unlikely to be the optimal price, and the revenue distributions at $p = 1$ and $p = M_1, M_2, \ldots, M_m$ are almost i.i.d.\@ when each $M_i \gg M_{i - 1}$ and $k \gg M_m$.
\end{itemize}

\subsection{Proof of \texorpdfstring{\Cref{thm:lower_bound}}{Proposition 4.3}}
\label{appendix:proof_of_lowerbound}

For any \(\delta > 0\) and \(c\in [\delta, 1]\), consider the distribution \(\dist\) defined as follows:  
\[
v \sim \dist, \quad v =
\begin{cases}
    2c & \text{with probability } \delta / 2c, \\
    1 & \text{otherwise}.
\end{cases}
\]

Observe that \(\dist\) has binary support, where the higher value is an integer multiple of \(c\). Applying \Cref{lem:binary_supp}, we conclude that with probability exactly \(\delta/c\), the empirical revenue of \(2c\) reaches \(c\) after a certain number of samples. Moreover, the welfare induced by the price \(2c\) is precisely \(\delta\). Thus, we have established that with probability exactly \(\delta/c\), the worst-case strategy selects a price with induced welfare less than \(\delta\), completing the proof.

\subsection{Proof of \texorpdfstring{\Cref{lem:discretize}}{Lemma 4.5}}
\label{appendix:proof_of_discretize}

Given an arbitrary (potentially non-discrete) distribution \(\dist\), we construct a discrete distribution \(\dist'\) by rounding up values of \(\dist\) to the nearest multiple of \(\varepsilon\). Formally, we define \( \tau \) as  
\[
\tau = \inf\{x \mid \wel(x) \leq \delta\}.
\]

Without loss of generality, we assume \(\wel(\tau) > \delta\). If instead \(\wel(\tau) \leq \delta\), the proof proceeds similarly with a modified discretization approach. Since \(\wel\) is non-increasing, it follows that for all \(x > \tau\), we have \(\wel(x) \leq \delta\).

Now we define the discretized distribution \(\dist'\) by rounding up the values above \(\tau\). 
\begin{align*}
    \Pr_{v\sim \dist'}[v = x] = \begin{cases}
        \Pr_{v\sim \dist}[v \leq \tau], & \text{if } x = 1, \\
        \Pr_{v\sim \dist}[x - \varepsilon < v \leq x], & \text{if } x = \tau + k\varepsilon ~\text{ for }~ k \in \mathbb{Z}_{+}, \\
        0, & \text{otherwise}.
    \end{cases}
\end{align*}

Define the welfare function with respect to the distribution \(\dist'\) as  
\[
\wel'\InParentheses{x} = \E_{\val \sim \dist'}\InBrackets{\val \cdot \indic{\val \geq x}}.
\]
By the definition of \(c\), it follows that  
\[
\wel'\InParentheses{\tau + \varepsilon} = (\tau + \varepsilon) \Pr_{\val \sim \dist}[\val > \tau] \leq \delta + \varepsilon.
\]
Similarly, since \(\wel'\InParentheses{x}\) is also non-increasing, any price greater than \(\tau + \varepsilon\) results in a welfare of at most \(\delta + \varepsilon\).

Next, we couple the sampling processes of \(\dist\) and \(\dist'\) as follows. Given a sample \(x\) from \(\dist\), if \(x \leq \tau\), we assume that the corresponding sample from \(\dist'\) is \(1\). Otherwise, we define  
\[
x' = \left\lceil \frac{x}{\varepsilon} \right\rceil \varepsilon
\]  
as the smallest multiple of \(\varepsilon\) that is at least \(x\). This construction ensures that sampling \(x'\) from \(\dist'\) is equivalent to the process described.

Define \(\empdistS{k}\) and \(\tilde{F}_s'^{(k)}\) as the complementary cumulative distribution functions of the empirical distributions after \(k\) samples. Suppose that after \(k\) samples, there exists a price \(x > \tau\) such that its empirical revenue satisfies  
\[
 x \cdot \empdistS{k}\InParentheses{x} \geq c.
\]

Next, define  
\[
x' = \left\lceil \frac{x}{\varepsilon} \right\rceil \varepsilon
\]  
as the smallest multiple of \(\varepsilon\) that is at least \(x\). By our coupled sampling procedure, the number of samples from \(\dist'\) that are at least \(x'\) must be no less than the number of samples from \(\dist\) that are at least \(x\). Consequently, we obtain  
\[
x' \cdot \tilde{F}_s'^{(k)}\InParentheses{x'} \geq x \cdot \empdistS{k}\InParentheses{x} \geq c.
\]

This implies that the probability of the existence of a price with induced welfare at most \(\delta + \varepsilon\) achieving an empirical revenue of at least \(c\) under \(\dist'\) is no less than the probability of the existence of a price with induced welfare at most \(\delta\) achieving an empirical revenue of at least \(c\) under \(\dist\). This completes the proof.
 
\end{document}